\documentclass{article}
\usepackage[margin=1in]{geometry}
\usepackage[fleqn]{amsmath}
\usepackage{amsthm,amsfonts,amssymb,bm}
\usepackage{authblk}
\usepackage{srcltx}

\usepackage{multirow}
\usepackage{graphicx,wrapfig,lipsum}
\usepackage{subfig}
\usepackage{float}
\usepackage{cite}
\usepackage{makecell}
\usepackage[labelfont=bf, skip=7pt]{caption}
\usepackage{enumitem}

\usepackage{amsthm}
\newtheorem{lem}{Lemma} 
\newtheorem{prop}{Proposition}

\newcommand{\Y}[1]{{(Y#1)}}

\title{Polynomial Approximations of Hysteresis Curves Near the Demagnetized State}

\author{S. E. Langvagen\thanks{Electronic address: sergey.langwagen@gmail.com}}

\affil{\small Chernogolovka, Moscow Region}

\date{November 21, 2017}

\begin{document}

\maketitle

\begin{abstract}
  Polynomial approximations of hysteresis curves were studied for
  systems exhibiting the return point memory. An extended Rayleigh law
  that uses polynomials of the third degree, and Rayleigh-like
  equations describing the energy dependence on the applied magnetic
  field are proposed. The results were compared with numerical
  experiments on a zero temperature random bond Ising model.
\end{abstract}

\section{Introduction}

Symmetric hysteresis loop and the virgin magnetization curve in
the neighborhood of the demagnetized state are described by the
equations
\begin{equation}\label{Rayleigh}
M = (a +bH_m)H \pm \frac{b}{2}(H^2 - H_m^2),\qquad M = aH \pm bH^2,
\end{equation}
where the upper and lower signs distinguish the ascending and
descending branches. Equations (\ref{Rayleigh}) represent the
so-called Rayleigh law \cite{Bozorth1951, Chikazumi1997, Bertotti1998,
  Cullity&Graham2009}, named after Lord Rayleigh, who discovered them
experimentally \cite{Rayleigh1887}. Rayleigh equations have been
confirmed for many ferromagnetic materials.  Neel gave the first
explanation of the Rayleigh law in terms of domain walls moving in a
random energy landscape \cite{Neel1942, Neel1943}. For recent
development in the microscopic foundation of the Rayleigh law see,
e.g., \cite{Zapperi&all2002} and references therein.

This work does not concern details of the underlying mechanism
responsible for the Rayleigh law.  Instead, restrictions on hysteresis
curves imposed by the return point memory, also called ``wiping out''
property, \cite{Bertotti1998, Mayergoyz2003, Sethna&all1993} are
studied. The consideration is based on the results of the previous
work \cite{Langvagen2017} that are summarized below for convenience of
the reader.

\bigskip

Let the slowly varying uniaxial magnetic field $H(t)$ is applied to a
demagnetized ferromagnetic specimen. Let $H$ decreases by $\Delta H_0$
starting from the value $H = 0$, then increases by $\Delta H_1$, then
decreases by $\Delta H_2$ and so on till $\Delta H_n$, as shown in
Fig. \ref{fig:HM_curves}.  The final macroscopic state of the specimen
is completely determined by the sequence $\Delta H_0, \ldots, \Delta
H_n$, where $n$ can be any number, $n = 0,1,\,\ldots$\;.

If the specimen exhibits the return point memory (RPM), all the states
that can be obtained by applying $H(t)$, can be reached by the process
such that
\begin{equation}\label{DH-ineq}
  2\Delta H_0> \Delta H_1 > \ldots \Delta H_n > 0.
\end{equation}
These values are considered as coordinates in the so called ``minimal
space of states'', which includes all and only the states
reachable from the demagnetized state by applying $H(t)$. 


It is convenient to designate
\begin{equation}\label{xi-def}
  \xi_0 = 2\Delta H_0,\, \xi_1 = \Delta H_1,\, \xi_2 = 
  \Delta H_2,\, \ldots,\, \xi_n = \Delta H_n
\end{equation}
and assume that
\begin{equation}\label{xi-ineq}
  \xi_0 \geq \xi_1 \geq \xi_2 \geq \ldots \geq \xi_n \geq 0.
\end{equation}
\begin{figure}[H]
  \begin{center}
    \includegraphics[scale=0.57]{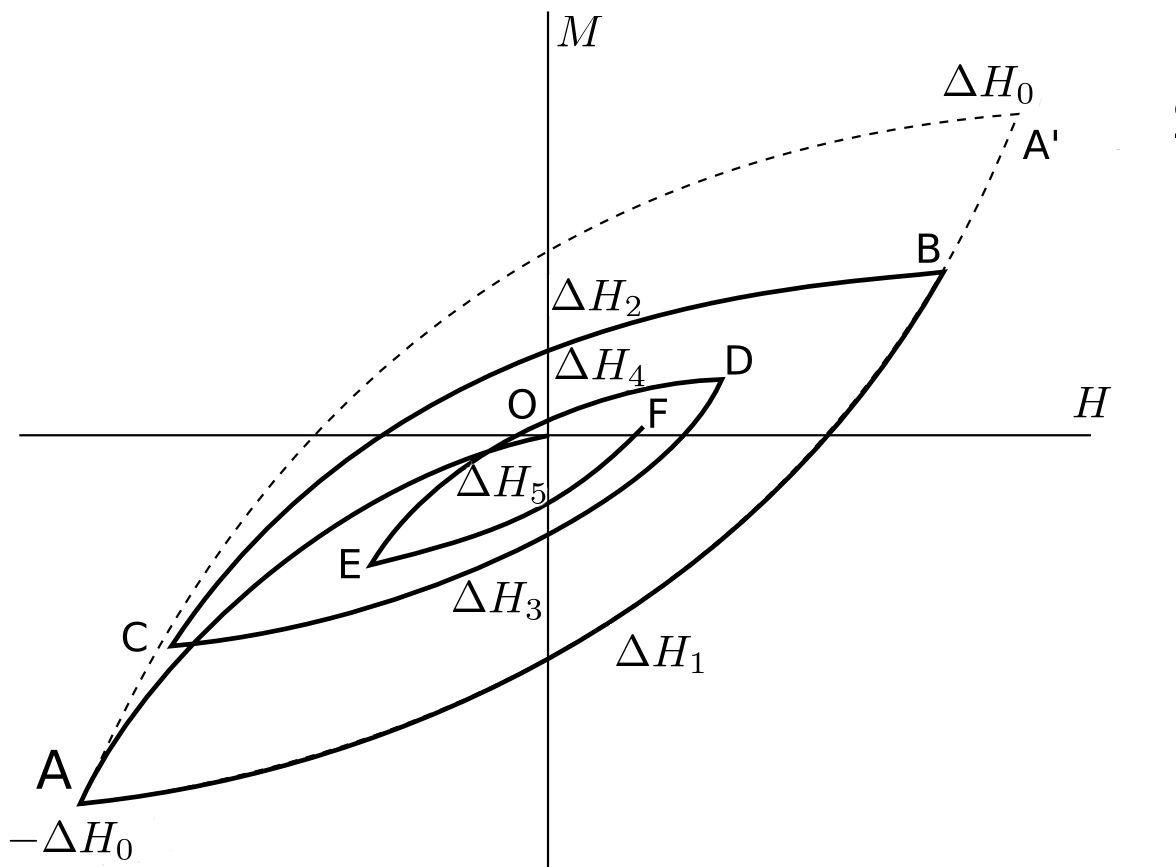}
    \caption{Magnetization process $OABCDEF$ starts from the
      demagnetized state at the point $O$ and is performed by
      decreasing the field $H$ by the value $\Delta H_0$, then
      increasing it by $\Delta H_1$, then decreasing it by $\Delta
      H_2$ and so on. The final state of the system is determined by
      the values $\Delta H_0, \ldots,\Delta H_n$, where $n+1$ is the
      number of hysteresis branches. The symmetric cycle $AA'$ is
      shown with dashed line.}\label{fig:HM-curves}
  \end{center}
\end{figure}\label{fig:HM_curves}
\vspace{-3mm}

The behavior of any related to the specimen macroscopic physical value
$y$ that depends on the magnetic state, can be expressed as a sequence
of functions
\begin{equation}\label{y(xi)-read-out}
  \{y_n(\xi_0,\ldots,\xi_n)\},\quad n = 0,1,\,\dots\,.
\end{equation}
Functions (\ref{y(xi)-read-out}) must satisfy the following
conditions:
\begin{equation}\nonumber
  \mbox{\Y0}\;\;\; y_n(\xi_0,\xi_1,\ldots,\xi_n) = 
  y_{n-1}(\xi_0,\xi_1,\ldots,\xi_{n-1})$,\;\mbox{ if }$\xi_n = 0,\,n\geq 1.
\end{equation}
\begin{flalign*}
  \mbox{\Y1}\;\;\;y_n(\xi_0,\ldots,\xi_k,\xi_{k+1},\ldots,\xi_n) = 
  y_{n-2}(\xi_0,\ldots,\xi_{k-1},\xi_{k+2},\ldots,\xi_n),\\
  \mbox{ if }\xi_k = \xi_{k+1},\,\, 1\leq k\leq n - 1,\, n\geq 2.
\end{flalign*}
\begin{equation}\nonumber
  \mbox{\Y2}\;\;\; 2\frac{\partial y_n}{\partial \xi_0} + 
  \frac{\partial y_n}{\partial \xi_1} = 0,\; 
  \mbox{if}\; \xi_0 = \xi_1,\, n = 1, 2, \ldots\,.
\end{equation}

\medskip

The first condition \Y0 seems to be obvious, and the second \Y1
directly follows from the RPM. Condition \Y2 gives the possibility to
obtain the demagnetized state by applying alternating magnetic field
of gradually decreasing amplitude. According to \cite{Langvagen2017},
it guarantees that the physical value described by the sequence of
functions $\{y_n(\xi_0,\ldots ,\xi_n)\}$ becomes equal to its value in
the initial demagnetized state $y_0(0)$ after the
demagnetization. Note that the sequences of functions that follow \Y0
-- \Y2 form a linear space.

If the state $(\xi_0,\xi_1,\xi_2,\,\ldots\,,\xi_n)$ is obtained from
the demagnetized state with the input $H(t)$ the state
$(\xi_0,\xi_0,\xi_1,\, \ldots\, ,\xi_n)$ can be obtained with the
input $-H(t)$.  For the magnetic hysteresis, we are usually interested
in {\em symmetric} or {\em antisymmetric} functions satisfying one of
the following conditions:
\begin{align}\nonumber
&\mbox{\Y{s}}\quad y_n(\xi_0, \xi_1, \xi_2, \ldots , \xi_n) = 
y_{n-1}(\xi_0,\xi_2,\ldots,\xi_n),\;\mbox{if}\;\xi_1 = \xi_0,\; 
 n = 1,2,\ldots\,, \nonumber\\[1em]
& \mbox{\Y{a}}\quad y_n(\xi_0,\xi_1,\xi_2\ldots,\xi_n) = 
-y_{n-1}(\xi_0,\xi_2,\ldots,\xi_n),\;\mbox{if}\;\xi_1 = \xi_0,\;
n = 1,2,\ldots\,.\nonumber
\end{align}
The antisymmetric functions can describe magnetization $M$, and the
symmetric ones can describe, for example, the energy of the specimen.

\section{Taylor expansion}

Below we assume that $y_n(\xi_0,\ldots,\xi_n)$ 
have continuous partial derivatives of sufficient order at any point
in the region determined by inequalities (\ref{xi-ineq}).

\medskip

If the sequence $\{y_n(\xi_0,\ldots,\xi_n)\}$ satisfies conditions \Y0
-- \Y2, and optionally \Y{s} or \Y{a}, these conditions also hold true
for the sequence $\{y_n(\lambda\xi_0,\ldots,\lambda\xi_n)\}$ with any
constant $\lambda > 0$. It is also not difficult to see that \Y0 --
\Y2, and \Y{s} or \Y{a} hold true for the sequence
\begin{equation}\label{y(lambda)-diff}
  \big\{\frac{d^k}{d\lambda^k}y_n(\lambda\xi_0,
  \ldots\,,\lambda\xi_n)\big\},
\quad
n = 0,1,2,\ldots\,,
\end{equation}
with fixed $k = 0, 1, 2, \ldots\,$ and $\lambda > 0$.

\medskip

According to Taylor's theorem
\begin{equation}\label{Taylor}
  y_n(\xi_0,\ldots ,\xi_n) = 
  \sum_{k=0}^r P_n^{(k)}(\xi_0,\ldots ,\xi_n) + r_n(\xi_0,\,\ldots\,\xi_n),
\end{equation}
where
\begin{equation}\nonumber
  P^{(k)}_n(\xi_0,\ldots,\xi_n) = \frac{1}{k!}
  \frac{d^k y_n(\lambda\xi_0,\ldots,\lambda\xi_n)}{d\lambda^k}\Big|_{\lambda = 0}\,,
\quad
r_n(\xi_0,\,\ldots\,\xi_n) = o(\xi_0^r).
\end{equation}
Here $P^{(k)}_n$ are homogeneous polynomials of degree $k$. The
estimate of the reminder $r_n$ is written taking into account
inequalities (\ref{xi-ineq}). The estimate is uniform with respect to
$n$ if derivatives (\ref{y(lambda)-diff}) of order $r+1$ are uniformly
bounded with respect to $n$.  As follows from (\ref{y(lambda)-diff})
when $\lambda$ tends to zero, the sequences of polynomials
$\{P^{(k)}_n(\xi_0,\ldots,\xi_n)\}$ must satisfy conditions \Y0 -- \Y2
and optionally \Y{s} or \Y{a}, for any $k = 0,1,2, \ldots\,$.

\medskip

If conditions \Y0 -- \Y2 are applicable, expansion (\ref{Taylor}) gives
polynomial approximation of corresponding hysteresis curves near the
demagnetized state.  Similar consideration can be preformed in the
neighborhood of any state with fixed coordinates $\xi_0,\ldots\xi_m$
by expanding functions $y_n(\xi_0,\ldots \xi_m,
\lambda\xi_{m+1}\ldots\lambda\xi_n)$, $0<\lambda\leq 1$. In this case
polynomials $P^{(k)}_{m+1,n}(\xi_{m+1}\ldots,\xi_n)$ must satisfy
conditions \Y0, \Y1 only, and likely have different coefficients at
different points $(\xi_0,\ldots ,\xi_m)$.

\medskip

In the following study, the consideration is restricted to the
neighborhood of the demagnetized state and to the polynomials of the
third degree and lower. Note that the terms {\em functions} and {\em
  sequence of functions}, {\em polynomials} and {\em sequence of
  polynomials} are used interchangeably.

\section{Elementary Homogeneous Polynomials}

Homogeneous polynomials up to the third degree that follow
conditions \Y0 -- \Y2 are listed in Table~\ref{tab:poly3-elem}, where
\begin{equation} \label{poly3-elem}
\sigma^{(p)}_n = \sum_{i=0}^{n} \epsilon_i \xi^{p}_i\,,\quad
\sigma^{(12)}_n = \frac{1}{12}\,\xi_0^3 + 
\frac{1}{2} \sum_{1\leq i \leq n} \xi^3_i + \!\!\!\!\!
\sum_{0 \leq i<j \leq n} \epsilon_i \epsilon_j \xi_i \xi_j^2\,,\quad
\sigma^{(21)}_n = \frac{1}{6}\,\xi_0^3 + 
\frac{1}{2} \sum_{1\leq i \leq n} \xi^3_i + \!\!\!\!\!
\sum_{0 \leq i<j \leq n} \epsilon_i \epsilon_j \xi_i^2 \xi_j\,.
\end{equation}
Here $p = 1,2,3$, and 
\begin{equation}\label{epsilons}
\epsilon_0=-1/2 ,\;\epsilon_i=(-1)^{i+1} 
\;\mbox{for}\; i=1,2,\,\ldots\,. 
\end{equation}

Taking into account (\ref{xi-def}), it can be seen that
$\sigma^{(1)}_n = H$. Polynomials $\sigma^{(1)}_n$,
$\big(\sigma_n^{(1)}\big)^2$, $\big(\sigma_n^{(1)}\big)^3$ are
unhysteretic.  Polynomials $\sigma^{(2)}_n$, $\sigma^{(3)}_n$
represent the ordinary and $\sigma^{(12)}_n$, $\sigma^{(21)}_n$ the
butterfly-shaped hysteresis curves, and the relation
$\sigma^{(1)}_n\sigma^{(2)}_n = \sigma^{(12)}_n + \sigma^{(21)}_n$
holds true. 

Conditions \Y0 -- \Y2 can be easily verified for all the polynomials
in Table~\ref{tab:poly3-elem}. We call these polynomials elementary
because, as follows from Proposition \ref{prop:poly3-elem}, they form
a basis in the linear space of the sequences of third-degree
polynomials satisfying conditions \Y0 -- \Y2.

\begin{table}[H]
\caption{Elementary Polynomials up to Degree 3} \label{tab:poly3-elem}
\setcellgapes{4pt}
\makegapedcells
\begin{center}
\begin{tabular}{|c|c|c|}
\hline
{\bf Degree} & {\bf Symmetric} & {\bf Antisymmetric}\\
\hline
$0$ & $1$ & $-$ \\
$1$ & $-$ & $\sigma^{(1)}_n$ \\
$2$ & $\big(\sigma^{(1)}_n\big)^2$ & $\sigma^{(2)}_n$ \\
$3$ & $\sigma^{(12)}_n,\;\sigma^{(21)}_n$
& $\big(\sigma^{(1)}_n\big)^3,\;\sigma^{(3)}_n$\\
\hline
\end{tabular}
\end{center}
\end{table}

\begin{lem}\label{lem:Y1-diff}
For any $\{y_n(\xi_0,\ldots,\xi_n)\}$ satisfying conditions \Y0 -- \Y2 it holds
\begin{equation}\label{Y1-diff}
  2^{\delta_{0i_s}}\frac{\partial^r y_n}{\partial\xi_{i_1}\ldots
  \partial\xi_{i_s}\ldots\partial\xi_{i_r}} + 
  \frac{\partial^r y_n}{\partial\xi_{i_1}\ldots\partial\xi_{i_s+1}
  \ldots\partial\xi_{i_r}} = 0\,,\quad
  \mbox{if}\quad \xi_{i_s+1} = \xi_{i_s},
\end{equation}
where $i_1<\ldots<i_s,i_s+1<\ldots< i_r$, $r=1,2,\,\ldots\,$,\; 
$s=1,\,\ldots\,,r$, and  $n = 1, 2\,\ldots\,$.
\end{lem}

\begin{proof}
  As shown in \cite{Langvagen2017}, conditions \Y1 and \Y2 can be
  combined in one:
\begin{equation}\label{Y1*}
  2^{\delta_{0k}}\frac{\partial y_n}{\partial \xi_k} + 
  \frac{\partial y_n}{\partial \xi_{k+1}} = 0\,,\; \mbox{if}\; 
  \xi_k = \xi_{k+1},\, n = 1,2,\ldots\,,
\end{equation}
where $\delta_{ij}$ is the Kronecker delta.  Because (\ref{Y1*}) is
true for an arbitrary $\xi_i$, $i \neq k, k+1$, it can be
differentiated by any $\xi_i$ any times giving (\ref{Y1-diff}).
\end{proof}

\begin{prop}\label{prop:poly3-elem}
  Any homogeneous polynomials $P^{(1)}_n, P^{(2)}_n, P^{(3)}_n$ of the
  degree $1,2,3$ that satisfy conditions \Y0 -- \Y2 can be represented
  as a linear combination of polynomials listed in Table
  \ref{tab:poly3-elem} as follows:
\begin{equation}\nonumber
P^{(1)}_n = \alpha_1\sigma^{(1)}_n,\qquad
P^{(2)}_n = \alpha_2\sigma^{(2)}_n + \beta_2\left(\sigma^{(1)}_n\right)^2, \qquad
P^{(3)}_n = \alpha_3\sigma^{(3)}_n + \beta_3\sigma^{(12)}_n + \gamma_3\sigma^{(21)}_n +
\delta_3\left(\sigma^{(1)}_n\right)^3,
\end{equation}
where the constants $\alpha_1,\,\ldots,\,\delta_3$ do not depend on $n$.
\end{prop}

\begin{proof}
  Consider the proof for $P^{(3)}_n$.

Any homogeneous polynomials $P_n^{(3)}$ can be expressed in the
  following form:
\begin{equation}\nonumber
P_n^{(3)}(\xi_0,\ldots\,,\xi_n)=\underbrace{\sum_{0 \leq i \leq n} a_i^{(3)}\xi_i^3}_{A}+
\underbrace{\sum_{0 \leq i<j \leq n}a_{ij}^{(12)}\xi_i\xi_j^2}_{B}
+\underbrace{\sum_{0 \leq i<j \leq n} a_{ij}^{(21)}\xi_i^2\xi_j}_{C}+
\underbrace{\sum_{0 \leq i<j<k \leq n}a_{ijk}^{(111)}\xi_i\xi_j\xi_k}_{D}.
\end{equation}
After applying to $P^{(3)}_n$ any of differential operators
\begin{equation}\nonumber
2^{\delta_{0k}}\frac{\partial^3}{\partial\xi_i\partial\xi_j\partial \xi_k} + 
  \frac{\partial^3}{\partial\xi_i\partial\xi_j\partial \xi_{k+1}}
\quad\mbox{such that}\;
i<j<k, 
\;\mbox{or}\;
i<k,k+1<j, 
\;\mbox{or}\;
k+1<i<j,
\end{equation}
the sums $A,B,C$ vanish, and in sum $D$ remain the
following:
\begin{equation}\nonumber
a^{(111)}_{ijk}+a^{(111)}_{i,j,k+1}\;\mbox{for}\;i<j<k;\quad
a^{(111)}_{ikj}+a^{(111)}_{i,k+1,j}\;\mbox{for}\;i<k,k+1<j;\quad 
a^{(111)}_{kij}+a^{(111)}_{k+1,i,j}\;\mbox{for}\;k+1<i<j. 
\end{equation}
According to Lemma~\ref{lem:Y1-diff}, it must be that
\begin{equation}\label{a111-constr}
a^{(111)}_{ijk}+a^{(111)}_{i,j,k+1}=0;\quad
a^{(111)}_{ikj}+a^{(111)}_{i,k+1,j}=0;\quad 
2^{\delta_{0k}}a^{(111)}_{kij}+a^{(111)}_{k+1,i,j}=0\,. 
\end{equation}
Starting from $a^{(111)}_{012}$ and increasing indices one by one such
that the inequalities $0\leq i<j<k\leq n$ remain true, any coefficient
$a^{(111)}_{ijk}$ in sum $D$ can be obtained.  This means that, due to
(\ref{a111-constr}), all $a^{(111)}_{ijk}$ are determined by the first
coefficient $a^{(111)}_{012}$.  On the other hand, equations
(\ref{a111-constr}) are satisfied for $a^{(111)}_{ijk} =
\epsilon_i\epsilon_j\epsilon_k$ because $2^{\delta_{0i}}\epsilon_i =
(-1)^{i+1}$ according to definition (\ref{epsilons}).  Therefore, sum
$D$ must be proportional to the sum with the coefficients
$\epsilon_i\epsilon_j\epsilon_k$,
\begin{equation}\label{sum_D}
D\, \propto \!\!\!\sum_{0 \leq i<j<k\leq n}\!\!\!\!\!
\epsilon_i\epsilon_j\epsilon_k\,
\xi_i\xi_j\xi_k\,.
\end{equation}
The sum on the right side itself does not agree with \Y0 -- \Y2 but is
contained in the polynomial $\big(\sigma_3^{(1)}\big)^3$. Therefore,
$D$ can be excluded by subtracting $\big(\sigma_3^{(1)}\big)^3$ with
appropriate multiplier $\delta_3$. Coefficients $a^{(111)}_{012}$ can
not depend on $n$ due to condition \Y0 for $P^{(3)}_n$, hence
$\delta_3$ does not depend on $n$.

Polynomials $P^{(3)}_n - \delta_3\big(\sigma ^{(1)}_n\big)^3$
include the sums of type $A, B, C$ only and satisfy conditions \Y1
-- \Y2. After applying to $P^{(3)}_n -
\delta_3\big(\sigma^{(1)}_n\big)^3$ any of operators
\begin{equation}\nonumber
2^{\delta_{0k}}\frac{\partial^2}{\partial\xi_i\partial \xi_k} + 
  \frac{\partial^2}{\partial\xi_i\partial \xi_{k+1}}
\quad\mbox{such that}\;
i<k, 
\;\mbox{or}\;
k+1<i ,
\end{equation}
sum $A$ vanishes, and in sums $B,C$ remain the following:
\begin{equation}\nonumber
2(a^{(12)}_{ik}+a^{(12)}_{i,k+1})\xi_k + 2(a^{(21)}_{ik}+a^{(21)}_{i,k+1})\xi_i
\quad\mbox{for}\;i<k;\quad 
2(a^{(12)}_{ki}+a^{(12)}_{k+1,i})\xi_k + 2(a^{(21)}_{ki}+a^{(21)}_{k+1,i})\xi_i
\quad\mbox{for}\;k+1<i,
\end{equation}
where $\xi_{k+1}$ was substituted with $\xi_k$.
According to Lemma~\ref{lem:Y1-diff}, it must hold
\begin{equation}\label{a12_a21-constr}
a^{(12)}_{ik} + a^{(12)}_{i,k+1} = 0\,,\quad
2^{\delta_{0k}}a^{(12)}_{ki} + a^{(12)}_{k+1,i} =0\,,\quad 
a^{(21)}_{ik} + a^{(21)}_{i,k+1} = 0\,,\quad 
2^{\delta_{0k}}a^{(21)}_{ki} + a^{(21)}_{k+1,i} =0\,. 
\end{equation}
From reasoning similar to that leading up to equation (\ref{sum_D}),
it follows that
\begin{equation} \label{sum_B_C}
B\propto\sum_{0 \leq i<j \leq n} \epsilon_i \epsilon_j \xi_i \xi_j^2,
\quad
C\propto\sum_{0 \leq i<j \leq n} \epsilon_i \epsilon_j \xi_i^2 \xi_j\,.
\end{equation}
Sums $B$ and $C$ can be excluded from $P^{(3)}_n -
\delta_3\left(\sigma^{(1)}_n\right)^3$ by subtracting
$\sigma^{(12)}_n$ and $\sigma^{(21)}_n$ with appropriate coefficients
$\beta_3$ and $\gamma_3$. Polynomials $P^{(3)}_n -
\delta_3\left(\sigma^{(1)}_n\right)^3 - \beta_3\sigma^{(12)}_n -
\gamma_3\sigma^{(21)}_n$ contain the sum of type $A$ only, which is
completely determined by coefficient $a^{(3)}_0$ and can be excluded
by subtracting $\sigma^{(3)}_n$ with the appropriate coefficient
$\alpha_3$, giving
\begin{equation}
  P^{(3)}_n - \alpha_3\sigma^{(3)}_n - 
 \beta_3\sigma^{(12)}_n  -  \gamma_3\sigma^{(21)}_n - 
 \delta_3\big(\sigma^{(1)}_n\big)^3 = 0.
\end{equation}
Here $\alpha_3, \beta_3, \gamma_3, \delta_3 $ do not depend on $n$,
because the first coefficients in sums $A,B,C,D$ can not depend on $n$
due to \Y0.  This proves the statement for $P^{(3)}_n$. For
$P^{(1)}_n$, $P^{(2)}_n$ the proof is similar.
\end{proof}

\section{The Rayleigh Region and Beyond}\label{sec:Rayleigh}

Antisymmetric polynomials of up to the second degree give the
following approximation of $M(\xi_0,\ldots,\xi_n)$:
\begin{equation}\label{Rayleigh2_M(xi)}
M(\xi_0,\ldots,\xi_n) = a\sigma^{(1)}_n + \frac{b}{2}\sigma^{(2)}_n.
\end{equation}
This equation describes any hysteresis branch in the neighborhood of
the demagnetized state. According to (\ref{Rayleigh2_M(xi)}), the
equations of any branch of hysteresis curves and of the initial
magnetization curve are
\begin{equation}\label{Rayleigh2_DM}
\Delta M = a\Delta H \pm \frac{b}{2}\big(\Delta H\big)^2,\quad
M = aH \pm H^2,
\end{equation}
where $\Delta M$ denotes change of the magnetization after the return
point; the upper sign corresponds to ascending and the lower one to
descending branches. The same formulation of the Rayleigh law for
hysteresis branches not necessary pertaining to a symmetric cycle can
be found in \cite{Neel1942}. For a symmetric hysteresis loop,
(\ref{Rayleigh2_M(xi)}) gives Rayleigh equations (\ref{Rayleigh}).

The third-degree approximation of
$M(\xi_0,\ldots ,\xi_n)$ with antisymmetric polynomials taken from
Table \ref{tab:poly3-elem} reads
\begin{equation}\label{Rayleigh3_M(xi)}
M(\xi_0,\ldots,\xi_n) = a\,\sigma^{(1)}_n + \frac{b}{2}\sigma^{(2)}_n +
a'\big(\sigma^{(1)}_n\big)^3 + \frac{b'}{4}\sigma^{(3)}_n .
\end{equation}
It has two additional terms with new coefficients $a'$ and $b'$.

The simplest way to obtain equations for branches of a symmetric
hysteresis loop is to substitute $\xi_0 = 2H_m$, $\xi_1 = H_m \pm H$
in $\pm M(\xi_0,\xi_1)$, and for branches of the initial magnetization
curve to substitute $\xi_0 = 2|H|$ in $\mp M(\xi_0)$.
Equation (\ref{Rayleigh3_M(xi)}) gives the following expressions for
branches of symmetric hysteresis cycles and for the initial
magnetization curve:
\begin{equation}\label{Rayleigh3_M}
M = aH \pm \frac{b}{2}\left[(H_m \pm H)^2-2H_m^2\right]
+ a'H^3 \pm\frac{b'}{4}\left[(H_m \pm H)^3 - 4H_m^3\right], \quad
M = aH \pm bH^2 + (a'+b')H^3.
\end{equation}
In these equations, the upper sign corresponds to the ascending and the
lower one to the descending branches, and $-H_m\leq H\leq H_m$.

\medskip

Consider the coefficients $a(H_m)$, $b(H_m)$ determined from a
symmetric hysteresis cycle via the maximum magnetization $M_m$ and the
remnant magnetization $M_r$ as follows:
\begin{equation}\label{Rayleigh_ab}
a(H_m) = \frac{M_m - 2M_r}{H_m}, \quad b(H_m) = \frac{2M_r}{H_m^2}. 
\end{equation}
In the Rayleigh region $a(H_m) = a$, $b(H_m) = b$.  With the
third-degree terms taken into account $a(H_m)$, $b(H_m)$ show
quadratic and linear dependence on $H_m\,$,
\begin{equation}\label{Rayleigh3_ab}
a(H_m) = a + \left(\,a' -\frac{b'}{2}\,\right)\;H_m^2,\ \quad
b(H_m) = b + \frac{3\,b'}{4}\,H_m. 
\end{equation}

\section{Energy transformations}\label{sec:energy}

It is well known that magnetization processes in ferromagnets are
accompanied by irreversible heat generation as well as by reversible
heat exchange. The later is known as the magnetocaloric effect, it can
be comparable by the value with the hysteresis losses
\cite{Bozorth1951}. For simplicity, the following consideration is
restricted to hysteresis systems without the magnetocaloric effect.
In general case, the results presented in this section are not
applicable to real ferromagnets.

\medskip

Let $E$ be the energy of a ferromagnetic specimen per unit volume
without the term $-HM$ responsible for the interaction with the
external magnetic field $H$. For the subsequent consideration, the
only fact that matters is that the energy landscape is rough, and $E$
has numerous local minima divided by energy barriers large in
comparison with $kT$. When the external field changes, the previously
stable energy minimum becomes unstable, and the domain structure of
the specimen makes an irreversible jump to another minimum, lowering
the total energy $E - HM$. If $H(t)$ changes slowly enough, the value
of $H$ can be considered as the same before and after the jump, and
hence
\begin{equation}\nonumber
\delta E - H \delta M \leq 0.
\end{equation}

The energy $E$ as a function of state can be approximated with
symmetric polynomials from Table~\ref{tab:poly3-elem} as follows:
\begin{equation}\label{Rayleigh2_E(xi)-undef}
E(\xi_0,\ldots,\xi_n) = E_0 + \alpha \big(\sigma^{(1)}_n\big)^2 + 
\beta\sigma^{(12)}_n + \gamma\sigma^{(21)}_n.
\end{equation}

For the derivatives of functions $\sigma^{(12)}_n$, $\sigma^{(21)}_n$ with
respect to the last argument $\xi_n$, it holds that
\begin{equation}\nonumber
\frac{\partial \sigma^{(12)}_n}{\partial \xi_n} = 
-\frac{1}{2}\xi_n^2 + 2\epsilon_n\xi_n\sigma^{(1)}_n,
\qquad 
\frac{\partial \sigma^{(21)}_n}{\partial \xi_n} = 
\frac{1}{2}\xi_n^2 + \epsilon_n\sigma^{(2)}_n\,,
\quad\mbox{where}\quad n\geq 1.
\end{equation}

Therefore, on the $n$-th branch for small $\delta\xi_n$ we have
\begin{equation}\label{Rayleigh2_dE(xi)}
\delta E = \frac{\partial E}{\partial \xi_n}\delta\xi_n =  
\left[
2\alpha\sigma^{(1)}_n\epsilon_n + 
\beta\left(-\frac{1}{2}\xi_n^2 + 2\epsilon_n\xi_n\sigma^{(1)}_n\right)
+\gamma\left(\frac{1}{2}\xi_n^2 + \epsilon_n\sigma^{(2)}_n\right)
\right]\delta\xi_n
\quad (n\geq 1).
\end{equation}

In the Rayleigh region, equations (\ref{Rayleigh2_M(xi)}) and
$H = \sigma^{(1)}_n$ give
\begin{equation}\label{Rayleigh2_HdM(xi)}
H\delta M = H \frac{\partial M}{\partial \xi_n}\delta\xi_n = 
\epsilon_n(a + b\xi_n)\sigma^{(1)}_n\delta\xi_n \quad (n\geq 1).
\end{equation}  

By neglecting the magnetocaloric effect, we can write for the heat
dissipation
\begin{equation}\label{E-balance}
\delta Q = H \delta M - \delta E \geq 0.
\end{equation}

For the system that exhibits the return point memory, the states
before and after completing a hysteresis cycle are the same, in
accordance with \Y1. Because of this, $\oint dE = 0$, and $\oint HdM =
\oint dQ$ for any closed hysteresis loop. On the $n$-th hysteresis
branch, by taking into account up to the third-degree terms
\begin{equation}\nonumber
  \Delta Q_n = \alpha'\Delta H_n + \beta'(\Delta H)^2_n + \gamma'(\Delta H)^3_n,
\end{equation}
where $n = 1,2,\ldots$. Coefficients $\alpha'$, $\beta'$, $\gamma'$ do
not depend on $\Delta H_n$, however, $\alpha'$ and $\gamma'$ can
depend on $\Delta H_0,\ldots, \Delta H_{n-1}$. In the approximation
considered, the term $\gamma'$ is independent of $\Delta H_0,\ldots,
\Delta H_{n-1}$. It also can not depend on $n$, because otherwise the
heat generation on branches of symmetric hysteresis cycles will be
different.  The return point can be made anywhere on the branch
$\Delta H_n$ forming, according to (\ref{Rayleigh2_DM}), the loop of
the area $b\,(\Delta H_n)^3/6$.  If $\delta Q \geq 0$, the
inequalities $0\leq \Delta Q_n \leq b\,(\Delta H_n)^3/6$ must hold
true. Because $\alpha'$, $\beta'$ do not depend on $\Delta H_n$, it is
possible only if $\alpha' = 0$, $\beta' = 0$. As the result we have
\begin{equation}\label{DQ_dQ(xi)}
  \Delta Q_n = \frac{b}{12}(\Delta H_n)^3, \quad \delta Q_n = 
  \frac{b}{4}\;\xi_n^2\,\delta\xi_n \quad (n\geq 1).
\end{equation}

Now coefficients $\alpha, \beta, \gamma$ in
(\ref{Rayleigh2_E(xi)-undef}) can be determined by using the energy
balance (\ref{E-balance}).

Substituting (\ref{Rayleigh2_dE(xi)}), (\ref{Rayleigh2_HdM(xi)}),
(\ref{DQ_dQ(xi)}) in (\ref{E-balance}) and comparing the terms gives
$\alpha = a/2$, $\beta = b/2$, $\gamma = 0$. Finally we have
\begin{equation}\label{Rayleigh2_E(xi)}
E(\xi_0,\ldots,\xi_n) = 
E_0 + \frac{a}{2}\,\big(\sigma^{(1)}_n\big)^2 + \frac{b}{2}\,\sigma^{(12)}_n.
\end{equation}
By letting $\xi_0 = 2H_m$, $\xi_1=H_m\pm H$ in $E(\xi_0,\xi_1)$, and
$\xi_0 = 2|H|$ in $E(\xi_0)$ the following equations can be obtained
for branches of symmetric hysteresis cycles and for the initial
magnetization curve:
\begin{equation}\label{Rayleigh2_HE}
E = E_0 +\frac{a}{2}H^2 + \frac{b}{3}H_m^3 - \frac{b}{4}(H_m \pm H)(H_m^2 - H^2),
\quad
E = E_0 + \frac{a}{2}H^2 \pm \frac{b}{3}H^3,
\end{equation}
where $-H_m \leq H\leq H_m$, the signs $\pm$ distinguish the branches
of increasing and decreasing $H$ respectively, and the energy $E_0$ is
the energy of the demagnetized state. As follows from
(\ref{Rayleigh2_HE}), the branch of symmetric hysteresis cycle and the
initial magnetization curve have the second order contact at the
points $\pm H_m$.

\bigskip

The other third-degree symmetric polynomials $\sigma^{(21)}_n$
represent the energy changes for the inverse Rayleigh hysteresis. In
this case we have
\begin{equation}\nonumber
M(\tilde{\xi_0},\ldots , \tilde{\xi_n}) = \sigma^{(1)}_n,
\qquad
H(\tilde{\xi_0},\ldots , \tilde{\xi_n}) = \tilde{a}\sigma^{(1)}_n -
\frac{\tilde{b}}{2}\sigma^{(2)}_n,
\end{equation}
where the variables $\tilde{\xi_0},\ldots , \tilde{\xi_n}$ are defined as
$
  \tilde{\xi_0} = 2\Delta M_0,\, \tilde{\xi_1} = \Delta M_1,\,
  \ldots,\, 
  \tilde{\xi_n} = \Delta M_n,
$ 
similar to (\ref{xi-def}), and 
\begin{equation}
\tilde{a} = \frac{H_m - 2H_c}{M_m}, \quad \tilde{b} = \frac{2H_c}{M_m^2}, 
\end{equation}\nonumber
similar to (\ref{Rayleigh_ab}). Arguments like those leading to
(\ref{Rayleigh2_E(xi)}) give
\begin{equation}\nonumber
E(\tilde{\xi_0},\ldots,\tilde{\xi_n}) = 
E_0 + \frac{\tilde{a}}{2}\,\big(\sigma^{(1)}_n\big)^2 - 
\frac{\tilde{b}}{2}\,\sigma^{(21)}_n.
\end{equation}

\section{Comparison with Experiments on RBIM}
\label{sec:RBIM}

The consideration performed in the previous sections is based on quite
general assumptions and must presumably agree with hysteresis models
that show the return point memory, have smooth hysteresis curves, and
can be demagnetized by gradual reduction of alternating magnetic
field. The most suitable for the experiments seem to be zero
temperature Ising hysteresis models. The random field Ising model
(RFIM) shows precise RPM \cite{Sethna&all1993}. Analytical and
numerical study of RFIM in the Rayleigh region was presented in
\cite{Dante&all2002,Zapperi&all2002, Colaiori&all2002}.  Energy
changes and dissipation in RFIM were considered in
\cite{Ortin&Goicoechea1998}.  In this work, the random bond Ising
model (RBIM), also called the spin glass Ising model
\cite{Vives&Planes1994, Katzgraber&all2003}, was selected for the
comparison.  Like many real ferromagnets, RBIM usually demonstrates
some deviations from the return point memory.

Only a small fraction of spins take part in magnetization processes in
low fields, and, for accurate experiments, the model must have a
relatively large total number of spins. Because of this, obtaining the
demagnetized state can be time consuming, and simple models and
algorithms are preferred.

It is assumed that the Ising spins are placed in a ring and interact
with each other if the distance between them is not greater than $r$.
The Hamiltonian of the model is defined as follows:
\begin{equation}\nonumber
{\cal H} = -\!\!\!\!\!\sum_{\substack{|d(i,j)|\leq r\\i\neq j}}
\!\!\! J_{ij}s_is_j - H\sum_i s_i,
\quad s_i = \pm 1, \; 1\leq i, j \leq N,\;
\end{equation}
where distance $|d(i,j)|$ is determined by the equations $d(i,j)
\equiv (i-j)\,(\mathrm{mod}\,N)$ and $-[N/2] < d(i,j) \leq [N/2]$;
coupling parameters $J_{ij}$ are assigned randomly in the interval
$J_0 -\Delta J < \sqrt{r}\,J_{ij} < J_0 + \Delta J$. The model is free
of edge effects, and any desirable even coordination number $2r$ can
be specified. The magnetization and the internal energy per site are
given by the equations
\begin{equation}\nonumber
E= -\frac{1}{N}\!\sum_{\substack{|d(i,j)|\leq r\\i\neq j}}\; 
\!\!\! J_{ij}s_is_j,
\quad M = \frac{1}{N}\sum_i s_i.
\end{equation}
It was always assumed that $\Delta J = 1$, because changing
proportionally $J_0$ and $\Delta J$ changes the scale along the
$H$-axis only.  Dominating interactions are of the ferromagnetic type
for positive $J_0$ and of the antiferromagnetic type for negative
$J_0$.
\begin{figure}[H]
  \subfloat{\includegraphics[height=0.41\textwidth]{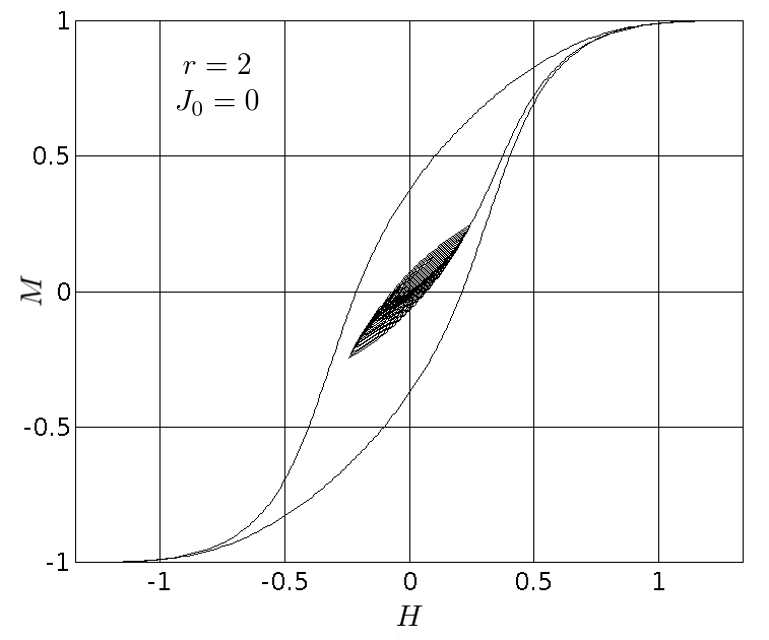}}
  \subfloat{\includegraphics[height=0.41\textwidth]{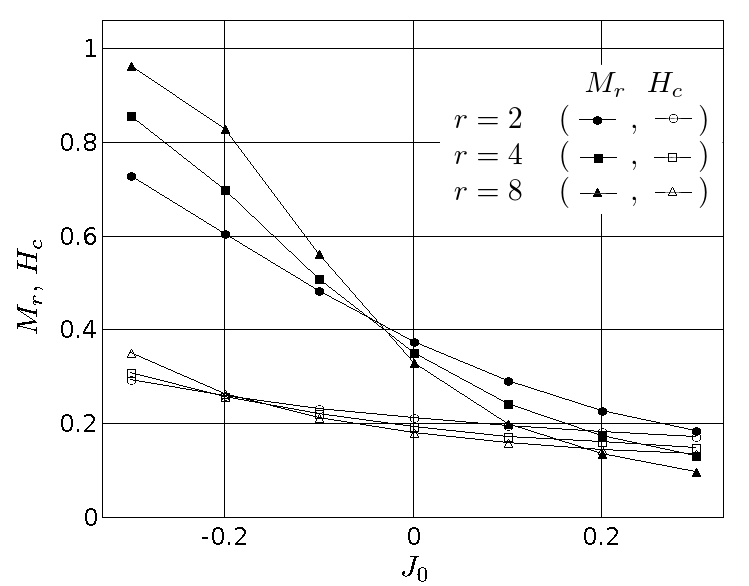}}
  \caption{Main hysteresis loop and a series of symmetric minor loops
    for $r = 2$, $J_0 = 0$ and $\Delta J = 1$ (left); Remanence and
    coercivity for different values of $r$, $J_0$, and $\Delta J = 1$
    (right).}
  \label{fig:RBIM_MH_MrHc}
\end{figure}
\noindent
The deterministic rules describing the dynamics of the model were
used. When the field changes, the stability of spins is checked in the
order of numbering. The first unstable spin flips and the neighboring
sites are updated and checked; again, the first unstable spin flips
and its neighboring sites are updated, and so on, until the spins in
the group become stable. Then the remaining spins are checked and
flipped in the same way, until all the spins are in the stable state.
Another dynamics with random selection between the unstable spins was
tested, with no noticeable difference in the shape of hysteresis
curves.

In the region $-0.3 \leq J_0 \leq 0.3$ the hysteresis curves are
comparable to those of ferromagnets, as shown in
Fig.~\ref{fig:RBIM_MH_MrHc}. The behavior of the model was studied
in this interval of $J_0$. The model demonstrates noticeable but not
very significant deviations from the macroscopic RPM. The deviations
from the microscopic RPM are as follows. For $J_0 = 0, r = 2$, RPM
holds with accuracy $0.2\%$ for $M_m = 0.2$, and with accuracy $4\%$
for $M_m = 0.6$.  Deviation from RPM increases with $r$; for $r=8$,
about $9\%$ of spins change orientation after completing the symmetric
hysteresis cycle with $M_m=0.6$.

\begin{figure}[H]
\centering
  \includegraphics[height=0.36\textwidth]{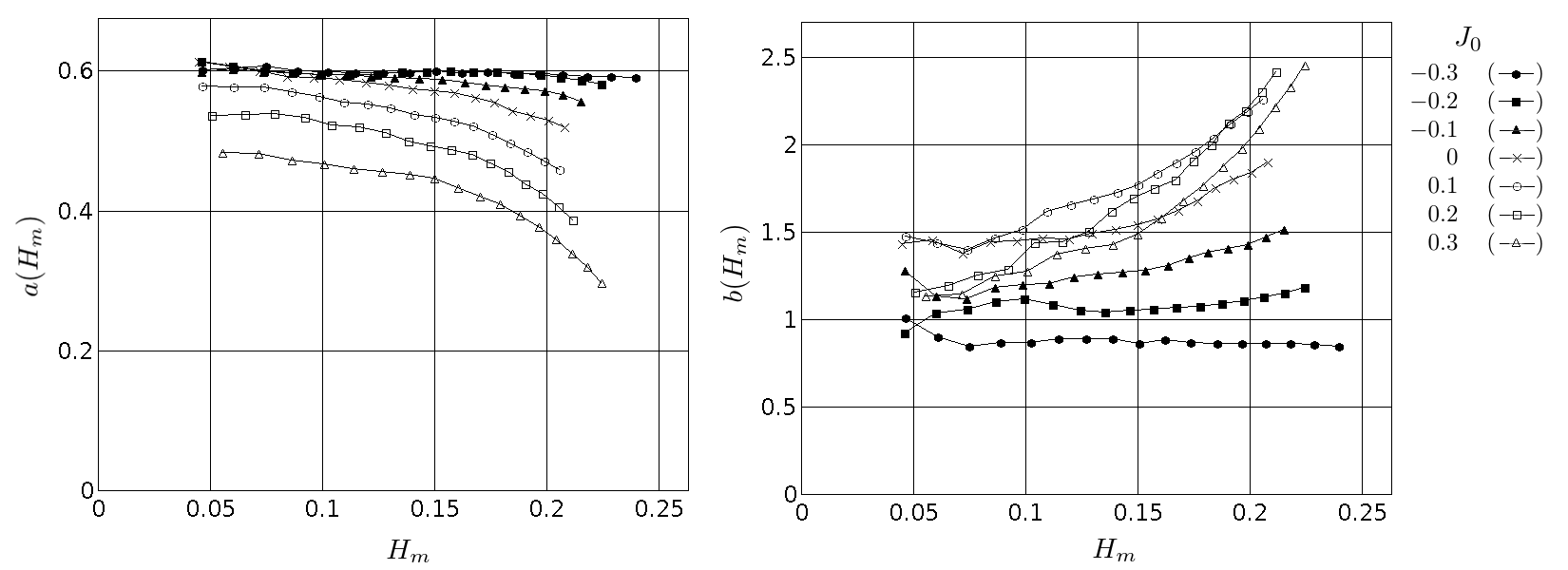}
  \caption{Rayleigh coefficients $a(H_m)$, $b(H_m)$ defined according
    to (\ref{Rayleigh_ab}). Parameters of the model are $\Delta J = 1$, $r
    = 2$, $J_0 = -0.3\,\ldots\,0.3$, $N = 2\cdot 10^5$.}
  \label{fig:RBIM_ab}
\end{figure}

For the experiments in the neighborhood of the demagnetized state were
taken $N = 2\cdot 10^5$, $r = 2$. The demagnetized state was obtained
by applying a series of cycles, each one with the maximum
magnetization equal to the maximum magnetization of the previous cycle
multiplied by a constant coefficient $k < 1$ selected close to
$1$. This procedure provides fine demagnetization near the
demagnetized sate, while for large $H$ the demagnetization is
relatively coarse.

Parameters $a(H_m),\,b(H_m)$ defined according to (\ref{Rayleigh_ab})
are presented in Fig.~\ref{fig:RBIM_ab}. Irregular behavior of the
curves could be explained by insufficient value of $N$, not fine
enough demagnetization, or imperfections of the random number
generator.  The irregular run of the curves in
Fig.~\ref{fig:RBIM_ab} do not allow to make a conclusion on
applicability of equations (\ref{Rayleigh3_ab}).

For the values of $H_m$ where the Rayleigh equations (\ref{Rayleigh})
hold true, $a(H_m)$, $b(H_m)$ must be equal to the Rayleigh constants
$a$, $b$. Not taking into account the irregularity of the curves in
Fig.~\ref{fig:RBIM_ab}, it can be expected that for $J_0 = -0.3$
and for $J_0 = -0.2$ the Rayleigh approximation (\ref{Rayleigh}) is
applicable, with some accuracy, up to $H_m = 0.2$.  It is unclear
whether the Rayleigh region is obtained or not for $J_0 = 0.2$ and
$J_0 = 0.3$.

\subsubsection*{Energy Transformations}

A consideration similar to that performed in Chapter~\ref{sec:energy}
can be applied to Ising spins, assuming that $Q$ denotes the energy
loss instead of the dissipated heat. Therefore, we can expect that
equation (\ref{Rayleigh2_HE}) holds true in the region of fields where
the Rayleigh law is applicable.

The curves with $H_m < 0.8$ were abandoned as not reliable. For $J_0 =
-0.2$, $J_0 = -0.3$, and $0.08 \leq H_m \leq 0.22$, equations
(\ref{Rayleigh}), (\ref{Rayleigh2_HE}) agree with the experiment as
shown in Fig.~\ref{fig:RBIM_HM_HE} as an example. The model gives
similar plots starting from $H_m \approx 0.08$ for all examined values
of $J_0$. While $H_m$ increases, the disagreement becomes noticeable
first with (\ref{Rayleigh2_HE}) and later with (\ref{Rayleigh}).  For
$J_0 = 0.3$, relatively small disagreement between (\ref{Rayleigh2_HE})
and the experiment is observed at $H_m = 0.1$, the disagreement with
(\ref{Rayleigh}) becomes apparent after $H_m = 0.15$.

\begin{figure}[H]
\centering
  \subfloat{\includegraphics[height=0.35\textwidth]{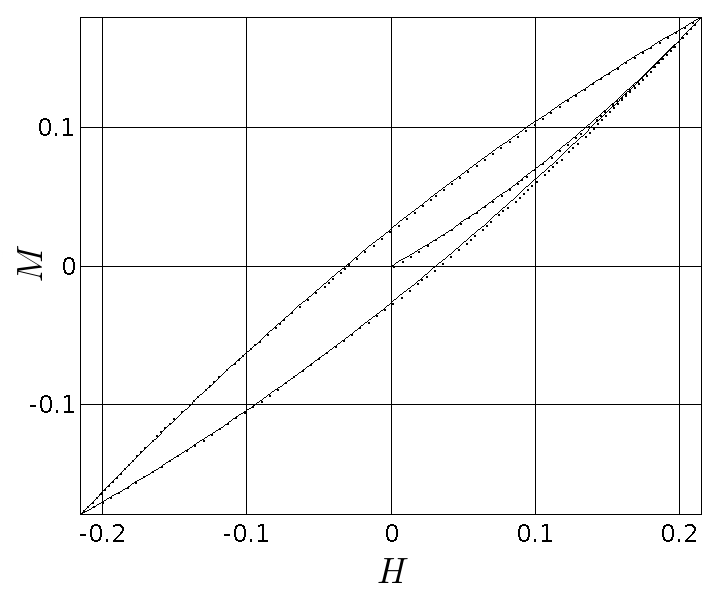}}
  \subfloat{\includegraphics[height=0.35\textwidth]{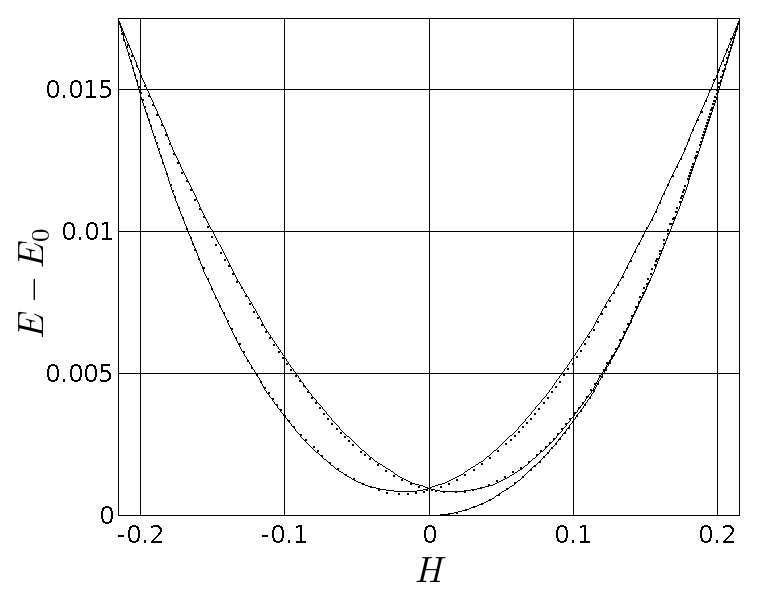}}
  \caption{Rayleigh hysteresis loop in $H, M$ and $H, E$ coordinates;
    $J_0 = -0.2$, $\Delta J = 1$, $r = 2$, and $H_m \approx 0.21$.
    Dotted curves represent the result of the numerical experiment on
    RBIM, solid ones are calculated according to equations
    (\ref{Rayleigh}), (\ref{Rayleigh2_HE}) with the same Rayleigh
    constants, by letting $a = a(H_m)$, $b = b(H_m)$.}
  \label{fig:RBIM_HM_HE}
\end{figure}

\section{Conclusions}

The sequences of polynomials $\{P^{(k)}_n(\xi_0, \ldots,\xi_n)\}$ that
are consistent with the return point memory and the reachability of
the demagnetized state must satisfy conditions \Y0 -- \Y2. These
polynomials can be used for Taylor expansion of the whole set of
hysteresis curves in the neighborhood of the demagnetized state. Eight
sequences of polynomials listed in Table~\ref{tab:poly3-elem} form a basis
in the linear space of the sequences of polynomials up to the third
degree.

There are only two antisymmetric polynomials up to the second degree
in the basis, and a linear combination of them~(\ref{Rayleigh2_M(xi)}) gives
the Rayleigh law.  Antisymmetric polynomials of the third degree add
two terms to the Rayleigh law according to equations
(\ref{Rayleigh3_M(xi)}), (\ref{Rayleigh3_M}).

Equation (\ref{Rayleigh2_E(xi)}) describing dependence of the energy
on the magnetic state $(\xi_0,\ldots ,\xi_n)$ were derived from the
following assumptions: (i) the hysteresis curves comply with the
Rayleigh law according to (\ref{Rayleigh2_DM}), and (ii) the heat is
always dissipated when the magnetic state changes. For symmetric
hysteresis cycles equation (\ref{Rayleigh2_E(xi)}) gives the
dependence of the energy on the applied magnetic field in the
Rayleigh-like form (\ref{Rayleigh2_HE}).

Equations (\ref{Rayleigh2_E(xi)}), (\ref{Rayleigh2_HE}) have no
adjustable parameters but are applicable only to hysteresis systems
without the magnetocaloric effect. In general case, they are not
applicable to real ferromagnets.  However, these equations must
presumably agree with hysteresis models that show the return point
memory, have smooth hysteresis curves, and can be demagnetized by an
alternating magnetic field. Numerical results obtained in the random
bond Ising model show reasonable agreement with
equation~(\ref{Rayleigh2_HE}).

\bibliography{hyst}{}
\bibliographystyle{plain}

\end{document}